%% file: q_trinity.tex
\documentclass[aps,twocolumn,prx,superscriptaddress,reprint, nofootinbib,floatfix]{revtex4-2}
\usepackage{etex}
\usepackage{graphicx}
\usepackage{hyperref}
\usepackage{physics}
\usepackage{float}
\usepackage{amsmath}

\usepackage{amsthm}
\usepackage{enumitem}
\usepackage{algcompatible}
\usepackage{algorithm}
\usepackage{marginnote}

\setlist{noitemsep}

\usepackage [
  n,
  advantage,
  operators,
  sets,
  adversary,
  landau,
  probability, 
  notions,
  logic,
  ff,
  mm,
  primitives,
  events,
  complexity,
  asymptotics,
  keys
  ] {cryptocode}

\usepackage{tikz, graphics}
\usetikzlibrary{arrows.meta, bending, patterns}
\usepackage{tikzscale} 

\floatname{algorithm}{Protocol}

\newcommand{\Z}{\mathsf{Z}}

\newcommand{\CZ}{\mathsf{CZ}}
\newcommand{\Correct}{\mathsf{Correct}}
\newcommand{\Redo}{\mathsf{Redo}}
\newcommand{\Abort}{\mathsf{Abort}}
\newcommand{\ok}{\mathsf{Ok}}
\newcommand{\NPC}{\mathsf{NP\!\! - \!\! Complete}}
\newcommand{\Id}{\mathbb{I}}

\newcommand{\pmax}{p_{\mathit{max}}}
\newcommand{\pmin}{p_{\mathit{min}}}

\newcommand{\fv}{\mathtt v}
\newcommand{\fV}{\mathtt V}
\newcommand{\fw}{\mathtt w}
\newcommand{\fW}{\mathtt W}
\newcommand{\ft}{\mathtt t}
\newcommand{\fT}{\mathtt T}
\newcommand{\fo}{\mathtt o}
\newcommand{\fO}{\mathtt O}

\newtheorem{theorem}{Theorem}
\newtheorem{lemma}{Lemma}
\newtheorem{corollary}{Corollary}
\newtheorem{definition}{Definition}

\hypersetup{
 pdfauthor={},
 pdftitle={Robust VBQC for Classical I/O},
 pdfkeywords={},
 pdfsubject={},
 pdfcreator={Emacs 26.3 (Org mode 9.3.1)}, 
 pdflang={English}}

\begin{document}

\date{\today}
\title{Securing Quantum Computations in the NISQ Era}
\author{Elham Kashefi}
\affiliation{School of Informatics, University of Edinburgh, 10 Crichton Street, Edinburgh EH8 9AB, United Kingdom}
\affiliation{Laboratoire d’Informatique de Paris 6, CNRS, Sorbonne Université, 4 Place Jussieu, 75005 Paris, France}
\author{Dominik Leichtle}
\affiliation{Laboratoire d’Informatique de Paris 6, CNRS, Sorbonne Université, 4 Place Jussieu, 75005 Paris, France}
\author{Luka Music}
\affiliation{Laboratoire d’Informatique de Paris 6, CNRS, Sorbonne Université, 4 Place Jussieu, 75005 Paris, France}
\author{Harold Ollivier}
\affiliation{Laboratoire d’Informatique de Paris 6, CNRS, Sorbonne Université, 4 Place Jussieu, 75005 Paris, France}

\begin{abstract}
Recent experimental achievements motivate an ever-growing interest from companies starting to feel the limitations of classical computing. Yet, in light of ongoing privacy scandals, the future availability of quantum computing through remotely accessible servers pose peculiar challenges: Clients with quantum-limited capabilities want their data and algorithms to remain hidden, while being able to verify that their computations are performed correctly. Research in blind and verifiable delegation of quantum computing attempts to address this question. However, available techniques suffer not only from high overheads but also from over-sensitivity: When running on noisy devices, imperfections trigger the same detection mechanisms as malicious attacks, resulting in perpetually aborted computations. Hence, while malicious quantum computers are rendered harmless by blind and verifiable protocols, inherent noise severely limits their usability.

We address this problem with an efficient, robust, blind, verifiable scheme to delegate deterministic quantum computations with classical inputs and outputs. We show that: 1) a malicious Server can cheat at most with an exponentially small success probability; 2) in case of sufficiently small noise, the protocol succeeds with a probability exponentially close to 1; 3) the overhead is barely a polynomial number of repetitions of the initial computation interleaved with test runs requiring the same physical resources in terms of memory and gates; 4) the amount of tolerable noise, measured by the probability of failing a test run, can be as high as 25\% for some computations and will be generally bounded by 12.5\% when using a planar graph resource state. The key points are that security can be provided without universal computation graphs and that, in our setting, full fault-tolerance is not needed to amplify the confidence level exponentially close to 1.
\end{abstract}

\maketitle

\section{Introduction}
\label{sec:intro}
\input{sec_files/01_introduction.tex}

\section{Measurement-Based Quantum Computing}
\label{sec:mbqc}
\input{sec_files/04_mbqc.tex}

\section{Noise-Robust Verifiable Protocol}
\label{sec:prot}
\input{sec_files/05_protocols.tex}

\section{Security Results and Noise Robustness}
\label{sec:sec_rob}

\subsection{Security Analysis}
\label{sec:corr_bl}
\input{sec_files/061_sec_res.tex}

\subsection{Noise Robustness}
\label{sec:noise_rob}
\input{sec_files/07bis_rob.tex}

\section{Discussion}
\label{sec:concl}
\input{sec_files/08_concl.tex}

\section{Conclusion}

We have presented above a protocol that is the first of its kind with respect to its ability to withstand noise while allowing to successfully execute computations in an unconditional verifiable way with relatively modest overhead.
These capabilities are made possible thanks to the nature of the computation, its deterministic classical output combined with a classical repetition code favourably replacing more resource-demanding fault-tolerant constructions.
The obtained error-correction capability can then be used to tolerate noise while having the computation perform correctly.

\subsubsection*{Acknowledgements}

We would like to thank Tracy Northup, Anders S{\o}ndberg S{\o}rensen, and Yuxiang Zhang for fruitful discussions. 
We acknowledge support of the European Union’s Horizon 2020 Research and Innovation Program under grant agreement number 820445 (QIA). 
DL gratefully acknowledges support from the European Union’s H2020 Program under grant agreement number ERC-669891 \mbox{(Almacrypt)}, and by the French ANR Projects ANR-18-CE39-0015 \mbox{(CryptiQ)} and ANR-18-CE47-0010 \mbox{(QUDATA)}.

\bibliographystyle{apsrev}
\bibliography{q_trinity}

\appendix

\section{Useful Inequalities from Probability Theory}
\label{app:tail_bounds}
\input{sec_files/app_tail_bounds.tex}

\section{Formal Security Definitions}
\label{app:form_def}
\input{sec_files/app_sec_def.tex}

\section{Proof of Perfect Local-Blindness}
\label{app:proof_lb}
\input{sec_files/069_local_blindness.tex}

\section{Proof of Verifiability}
\label{app:proof_ver}
\input{sec_files/07_verif.tex}

\section{Proof of Noise-Robustness}
\label{app:proof_rob}
\input{sec_files/app_proof_rob.tex}

\end{document}

%% file: sec_files/01_introduction.tex
Quantum computing promises unparalleled power for solving certain problems such as database search~\cite{grover} or integer factoring~\cite{shor}.
Recent experimental progress showed that the limit of classical un-simulatability is now within reach, if not already surpassed~\cite{google_supremacy}.
In this regime, quantum computers become so powerful that their classical counterparts cannot simulate their computation in a reasonable time.

On the one hand, this has triggered a lot of interest from all stakeholders starting to feel the limitations of classical computing power and wanting to prepare for the inevitable slow-down of Moore's law, i.e.~from academic labs all the way to industry users.
This, in turn, has driven most recent algorithmic and software developments in the field. More and more use cases are being studied with the goal of running useful computations on these devices as soon as their capabilities allow it.

On the other hand, because the cloud is emerging as the preferred way of accessing quantum computing capabilities, the questions of data and algorithm confidentiality as well as computation integrity are taking a particularly acute importance.
First, it is expected that only the most crucial and strategic computations will be run on these quantum computers, thus making these systems ideal targets for sophisticated hacking.
Second, disruption caused by (un)intentional mis-computations could stay undetected in absence of means to check the result.
In case of detection, it could still be difficult to pin-point the failing component due to the impossibility of following exactly the progress of quantum computations.
Several methods for eschewing this hurdle have been devised in the past (see e.g.~\cite{bfk, fk} and~\cite{GKK19} for a review).

In spite of the answers brought by the existence of these protocols, the initial questions are far from being resolved, even from a theoretical standpoint.
This is because currently known verifiable protocols are too sensitive to be of practical use.
Indeed, they have been developed for noiseless devices and have been optimized to detect the smallest fiddling and abort quickly.
Unfortunately, replacing perfect devices by even slightly noisy ones is not an option: the verification procedure would keep aborting, mistakenly thinking that plain imperfections are in fact the signature of malicious behaviour.

Several options for dealing with this sensitivity have been discussed in the past.
Previous research explored giving up blindness~\cite{GHK18}, imposing restrictions on the noise model~\cite{KD19}, moving towards the double server setting with classical clients and no classical communication between the servers~\cite{MF13a}, or finally going to the computationally-secure regime~\cite{Urm18}.
Yet, to obtain exponential security, these works would still require fully fault-tolerant computations, thus making them impractical for NISQ era devices.

As a consequence, before the availability of large fault-tolerant machines, clients could do their best to mitigate the effect of noise on their computations but, regarding security,  would still be left with no better alternative than to either give up their objectives entirely or try to convince themselves that providers are not as malicious as they could be. 
By running computations whose outcomes are known in advance, they could benchmark the performance and quality of the devices used (see e.g.~\cite{KLRB08a, WBDA+19a, MPJM19a}), and ultimately decide to trust or not future runs of the service.
Yet, such strategy falls short of the security expectations of most users because no benchmark entails future fulfilment of the provider's promises: benchmarking is not verifying.
It might only serve \textit{a posteriori} to demonstrate that the provider cheated during the benchmark, but would not help preventing a deviant behaviour at the time of computation.

In this paper, we propose a solution to both error mitigation and security issues by introducing a verifiable, blind and delegated quantum computing protocol for deterministic computations with classical inputs and outputs, that is also robust to noise (Section~\ref{sec:prot}).
It relies on the Measurement-Based model of Quantum Computation (MBQC) as it is the most natural one for delegation (see Section~\ref{sec:mbqc} for a short introduction to MBQC and its blind and trappified versions).
The robustness of the protocol means that there is no need to give up the ambition to provide security in a fully malicious adversarial model because of noise.
More precisely, we show in Section~\ref{sec:corr_bl} that, without any restriction on the adversary, any attempt at disturbing the computation will either be caught or error-corrected and made ineffective with high probability.
Hence, in the presence of noise that can be error-corrected by our scheme, the computation will be accepted with probability exponentially close to $1$.

To do this, we will repeat the computation several times while interleaving these executions with test runs aimed at detecting a dishonest behaviour of the Server.
The interplay between computations and tests will turn in our favour as it will both offer noise-robustness and exponential security.
The robustness will come from the error-correcting capability offered by majority voting, while the exponential security results from the blindness of the scheme which forces the Server to attack at least half the runs to have a chance to corrupt the computation. This in turn necessarily heightens its chances of getting caught.

We want to stress here that we make {\em no} assumption on the adversary in the process.
The adversary can be as malicious as it wants and the security will not be compromised.
Indeed, the security is information theoretic and provided in the composable framework of Abstract Cryptography (see~\cite{MR11,DFPR14} for an introduction).
It ensures it will not be jeopardised by subsequent or simultaneous instantiations in conjunction with other protocols.
Any noise is treated as an adversarial deviation, which is detectable by our cryptographic tool.
This is a fundamental property that cannot be obtained via any other certification approach~\cite{EHW+20} without using a masking scheme providing blindness. 

The key technical element used in this work is the recognition that, for deterministic classical-input classical-output computations, fault-tolerance is not required to error-correct the lightest attacks, thus allowing to boost the probability of acceptance of the protocol for noisy devices while still rightfully aborting when the perceived disturbance would risk overwhelming the classical error-correcting capabilities of the scheme.
The practical implication is the absence of overhead for each run used in the protocol when compared to the same non-robust, non-verifiable, blind delegated quantum computation in the MBQC model.
In fact, the only overhead of our scheme is to require the repetition of computations similar to the unprotected one (i.e.~same size, connectivities and gate set) a polynomial number of times in order to get not only exponential security, but also exponential acceptance when the noise is not too strong.
In particular, it does not increase the quantum memory requirement nor require additional simultaneous entanglement between quantum systems, thus being ideal for near-term implementations where each run could be carried out either sequentially on a single machine or in parallel using several ones.

%% file: sec_files/04_mbqc.tex
This section provides a brief overview of various useful notions linked to our model of computation, namely the Measurement-Based Quantum Computing (MBQC) and its blind and verifiable variants. Based on the gate-teleportation principle, this model is equivalent to the circuit model and is a natural setup for considering delegated computations, i.e.~when a Client with limited quantum capabilities instructs a more powerful Server to perform a computation on his behalf. We assume familiarity with the main other concepts of quantum information (see e.g.~\cite{nielsenchuang} for details) and refer the reader to \cite{mbqc} for a more in-depth introduction to MBQC.

From now on, we use the following notations. $\Z_{\theta}$ denotes the operator $\begin{pmatrix} 1 & 0 \\ 0 & e^{i\theta} \end{pmatrix}$ for $\theta \in \Theta = \{k\pi/4 \mid 0 \leq k \leq 7\}$ while $\ket{+_{\theta}}$ is the state $\Z_{\theta}\ket{+} = \frac{1}{\sqrt{2}}(\ket{0} + e^{i\theta}\ket{1})$. 

\subsection{Basic MBQC Definitions}
The basic MBQC Protocol for classical inputs and outputs can be summarized as follows. Any computation chosen by the Client is first translated into a graph $G = (V,E)$, where two vertices sets $I$ and $O$ define input and output vertices, and a list of angles $\{\phi_v \}_{v \in V}$. The set $\qty{G, I, O, \{\phi_v\}_{v \in V}}$ is called a measurement pattern.

To run a computation, the Client instructs the Server to prepare the graph state $\ket{G}$: for each vertex in $V$, the Server creates a qubit in the state $\ket{+}$ and performs a $\CZ$ gate for each pair of qubits forming an edge in $G$. The Client then asks the Server to measure each qubit of $V$ along the basis $\qty{\dyad{+_{\phi'_v}}, \dyad{-_{\phi'_v}}}$ in the order defined by the flow of the computation. The corrected angle $\phi'_v$  is given by $\phi'_v = (-1)^{s_v^X}\phi_v + s_v^Z\pi$ for binary values of $s_v^X$ and $s_v^Z$ that depend only on the outcomes of previously measured qubits and the flow. More details about the flow and the update rules for the measurement angles can be found in~\cite{hein2004multiparty,DK2006}.

\subsection{Hiding the Computation}
The computation can be totally hidden by using the following observation: if, instead of the Server preparing each qubit in the graph in the state $\ket{+}$, the Client sends $\ket{+_{\theta_v}}$ with $\theta_v \in_R \Theta$, where $\in_R$ indicates that a value was sampled from a set uniformly at random, then measuring the qubits in a similarly rotated basis, obtained by adding $\theta_v$ to the measurement angle, has the same result as the initial computation. If the Client keeps the angle $\theta_v$ hidden from the Server, the Server is completely blind on what computation is being performed. The angle $\theta_v$ acts as a One-Time Pad for the measurement angle $\phi'_v$. Nevertheless, because the Server could always measure the received qubits, it would still learn $1$ bit of information about the angle $\theta_v$, which can take $8$ values and so consists of $3$ bits. To prevent this, another parameter $r_v$ is added for each qubit. The resulting measurement angle sent to the Server is then  $\delta_v = \phi'_v(\phi_v, s_v^X, s_v^Z) + \theta_v + r_v\pi$, for $\phi'_v$, $s_v^X$, $s_v^Z$ defined as above. The parameter $r_v$ serves as a One-Time-Pad for the measurement outcome: if $b_v$ is now the outcome returned by the Server, then we have $s_v = b_v \oplus r_v$ with $s_v$ defined as above. In short, the Client sends randomly rotated qubits that appear maximally mixed to the Server and that become the resource state once entangled. It then guides the computation with a set of classical instructions that are adapted to the effectively prepared resource state but still look perfectly random to the Server. It is the combination of these two parts (quantum state preparation and classical instructions) that leads to the desired blind computation. This idea was first formalized in the \emph{Universal Blind Quantum Computation} (UBQC) Protocol in~\cite{bfk}.

This technique was used previously in protocols that also imposed the computation to be embedded in a universal graph such as brickwork graphs or dotted-complete graphs \cite{KW15,fk}. This last requirement however caused a blow-up in the number of qubits since the Client could not choose the most optimal graph for its desired computation but had to make it fit these universal graphs. Relaxing this requirement allows us to work directly with the same graph as the one used for the Client's desired computation rather than an expanded one. While this leaks the information about the underlying computation graph, all the parameters (i.e.\ angles) remain hidden, which turns out to be sufficient for blindness. The result is a drastic reduction of required memory qubits on the Server's side.

\begin{figure*}[t]
  	\centering
	\includegraphics[width=0.8\textwidth]{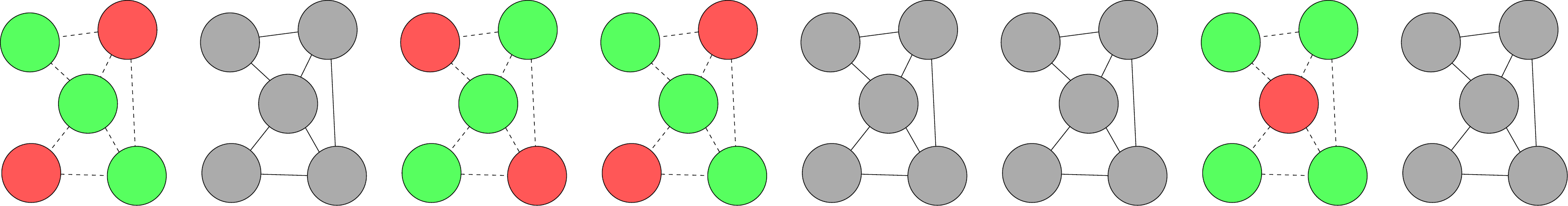}
	\caption{An example set of runs of the proposed protocol. Graphs in grey denote computation runs while graphs containing red nodes (traps) and green nodes (dummies) are test runs. This example graph on five nodes is completely covered with traps by the presented three types of test runs. Note that the Server remains completely oblivious of the differences between the runs, which are solely known to the Client.}
  	\label{fig:protocol_highlevel}
\end{figure*}

\subsection{Verifiability Through Trap Insertion}
In UBQC, the Server is not forced to follow the instructions and the Client can not verify if the computation is done correctly, but a modified version of the protocol allows for such verification. The central idea is to include \emph{trap qubits} at positions unknown to the Server~\cite{fk}. The conditions for successfully inserting traps are that they should have deterministic outcome if measured in the correct basis, should remain undetectable by the Server, and should not affect the computation.

To do this, the Client can send \emph{dummy qubits}, meaning qubits randomly initialized in states $\{\ket{0}, \ket{1}\}$ instead of the usual $\ket{+_{\theta}}$. This has the effect of breaking the graph at this vertex, removing it from $G$ along with any attached edges. Sending such dummies for all neighbours of a vertex isolates it from the rest of the graph, creating a trap. If measured in the same basis as the one used for their preparation, these traps yield deterministic outcomes while being undetectable by the Server. The latter is due to the fact that dummies appear as maximally mixed qubits from the Server's perspective and are thus no different than regular randomly chosen $\ket{+_{\theta}}$ states. The last condition to satisfy -- the possibility to still run the initial computation -- is more challenging. Without further assumptions, it requires the traps to be inserted within a modified graph that contains the computation to be performed. This results in a supplementary overhead in terms of stored qubits and applied gates compared to the UBQC version of the computation. The first protocol achieving verification through trappification was introduced in~\cite{fk} and achieved a quadratic overhead in the number of qubits. It was later optimized, such as in the \emph{Verifiable Blind Quantum Computation Protocol} (or VBQC) of~\cite{KW15} or in \cite{XTH20} to achieve a linear overhead.

The specificity of this paper is its focus on classical inputs and classical outputs. While this case covers the majority of current use cases of quantum devices, its main virtue is to allow for simple trap insertion strategies consisting of interleaving pure computation runs (i.e.\ without inserted traps) with pure test runs (i.e.\ only made up of traps). This simplicity will in turn lead to a simple amplification scheme yielding the desirable exponential confidence in the result of the verification procedure compared to the previous protocols that required full fault-tolerance universal schemes to achieve this level of confidence, thus blowing up the overhead.

More precisely, given a UBQC computation defined by a graph $G$, we construct \emph{test runs} based on a $k$-colouring $\{V_i\}_{i \in [k]}$ of $G$. Recall that by definition, a $k$-colouring satisfies
\begin{align*}
\bigcup_{i=1}^{k} V_i = V , \text{ and }\forall i \in [k],  \, \forall v \in V_i : N_G(v) \cap V_i = \emptyset,
\end{align*}
where $N_G(v)$ are the neighbours of $v$ in $G$. Hence, for each colour $i$, the Client could decide to insert traps for all vertices of $V_i$ while placing dummies in all other positions. This defines the test run associated to colour $i$. It is easy to check that the traps inserted in this way are isolated from other qubits, thus giving deterministic outcomes when measured in their preparation basis, and that they are undetectable from the Server as a test run results for the Server in applying the same sequence of operations as for the regular UBQC computation. The next section will describe how test runs can be used in conjunction with \emph{computation runs}, i.e.~regular UBQC computations, to achieve verified computation.

%% file: sec_files/05_protocols.tex
\begin{figure}
\begin{algorithm}[H]
\caption{Noise-Robust VBDQC with Deterministic Classical Output}
\label{prot:MQ-VBQC}
\begin{algorithmic}[0]
\STATE \textbf{Inputs:} The Server has no input. The Client has as input:
\begin{itemize}
\item Angles $\qty{\phi_v}_{v \in V}$ and flow $f$ on graph $G$.
\item Classical input to the computation $x \in \bin^{\#I}$.
\end{itemize}
\STATE \textbf{Protocol:}
\begin{enumerate}
\item The Client chooses uniformly at random a partition $(C, T)$ of $[n]$ ($C \cap T = \emptyset$) with $\#C = d$, the sets of indices of the computation and test runs respectively.

\item[2.] For $j \in [n]$, the Client and the Server perform the following sub-protocol (the Client may send message $\Redo_j$ to the Server before step 2.c while the Server may send it to the Client at any time, both parties then restart run $j$ with fresh randomness):
\begin{enumerate}

\item[(a)] If $j \in T$ (test), the Client chooses uniformly at random a colour $\mathsf{V}_j \in_R \qty{V_k}_{k \in [K]}$ (this is the set of traps for this test run).
\item[(b)] The Client sends $\#V$ qubits to the Server (where $\#X$ is the size of $X$). If $j \in T$ and the destination qubit $v \notin \mathsf{V}_j$ is a non-trap qubit (therefore a dummy), then the Client chooses uniformly at random $d_v \in_R \bin$ and sends the state $\ket{d_v}$. Otherwise, the Client chooses at random $\theta_v \in_R \Theta$ and sends the state $\ket{+_{\theta_v}}$.
\item[(c)] The Server performs a $\CZ$ gate between all its qubits corresponding to an edge in the set $E$.
\item[(d)] For $v \in V$, the Client sends a measurement angle $\delta_v$, the Server measures the appropriate corresponding qubit in the $\delta_v$-basis, returning outcome $b_v$ to the Client. The angle $\delta_v$ is defined as follows:
\begin{itemize}
\item If $j \in C$ (computation), it is the same as in UBQC, computed using the flow and the computation angles $\qty{\phi_v}_{v \in V}$. For $v \in I$ (input qubit) the Client uses $\tilde{\theta}_v = \theta_v + x_v\pi$ in the computation of $\delta_v$.
\item If $j \in T$ (test): if $v \notin \mathsf{V}_j$ (dummy qubit), it is chosen uniformly at random from $\Theta$; if $v \in \mathsf{V}_j$ (trap qubit), the Client chooses uniformly at random $r_v \in_R \bin$ and sets $\delta_v = \theta_v + r_v\pi$.
\end{itemize}
\end{enumerate}

\item[3.] For all $j \in T$ (test run) and $v \in \mathsf{V}_j$ (traps), the Client verifies that $b_v = r_v \oplus d_v$, where $d_v = \bigoplus_{i \in N_{G}(v)} d_i$ is the sum over the values of neighbouring dummies of qubit $v$. Let $c_{\mathit{fail}}$ be the number of failed test runs (where at least one trap qubit does not satisfy the relation above), if $c_{\mathit{fail}} \geq w$ then the Client aborts by sending message $\Abort$ to the Server.

\item[4.] Otherwise, let $y_j$ for $j \in C$ be the classical output of computation run $j$ (after corrections from measurement results). The Client checks whether there exists some output value $y$ such that $\# \left\{ y_j \, | \, j \in C,\, y_j = y \right\} > \frac{d}{2}$. If such a value $y$ exists (this is then the majority output), it sets it as its output and sends message $\ok$ to the Server. Otherwise it sends message $\Abort$ to the Server.
\end{enumerate}
\end{algorithmic}
\end{algorithm}
\end{figure}

While our Noise-Robust VBQC Protocol is formally defined in Protocol \ref{prot:MQ-VBQC}, we want to introduce it more intuitively in the next paragraphs in order to emphasize the features that make it suitable for practical purposes.

We now suppose that the Client has settled to perform a fixed computation which it has translated into a computational measurement pattern to be run on a graph $G$. It has also chosen a colouring $\{V_i\}_{i\in [k]}$ of $G$ and has sent it to the Server along with the graph $G$. 

The Client will now run the UBQC Protocol $n$ times successively, but with different update rules for the measurement angles. For $d$ of the runs chosen at random, the Client will update the measurement angles according to the computational measurement pattern thus resulting in computation runs. The remaining $t := n - d$ runs will be turned into test runs. More precisely, for each test run, the Client will secretly choose a colour at random and send traps for vertices of that colour and dummies everywhere else. As stated earlier, the trap qubits will be isolated from each other by the dummies. The Client then instructs the Server to measure all qubits as in computation runs, but with the measurement angle of trap qubits corresponding to the basis they were prepared in and a random measurement basis for the dummies. A test run is said to have \emph{passed} if all the traps yield the expected measurement results, and is said to have \emph{failed} otherwise. Figure~\ref{fig:protocol_highlevel} depicts one possible such succession of computation and test runs.

A direct consequence of this construction is that all runs share the same underlying graph $G$, the same order for the measurements of qubits, and all angles are chosen from the same uniform distribution. We will prove formally later that this implies blindness -- i.e.\ the Server cannot distinguish computation and test runs, nor tell which qubits are traps -- which in turn makes this trap insertion strategy efficient to obtain verifiability.

At the end of the protocol, the Client counts the number of test runs where at least one trap measurement has failed. If this number is higher than a given threshold $w$, the Client aborts the protocol by sending the message $\Abort$ to the Server.\footnote{$w$ would typically be set by the Client given its \textit{a priori} understanding of the quality of the Server. As discussed in Section~\ref{sec:noise_assump}, this does not affect security: a higher value would induce more runs than necessary to achieve a given confidence level, while a lower value would risk aborting with high probability.} Otherwise it sets the majority outcome of the computation runs as its output and sends message $\ok$ to the Server.

The parameters $n$, $d$, $t$, $w$ defined above are fixed for a given instantiation of the protocol and are publicly available to both parties in addition to the graph $G$ and the colouring $\{{V_i}\}_{i\in [k]}$. Their influence on the security bounds and on the noise-robustness of our protocol are detailed in the next sections along with the constraints they must abide. 

\paragraph{Redo Feature.}
Because the Client or the Server may experience failures in their experimental system, they might wish to discard and redo a run $j \in [n]$. In such situation, one of the parties can send a $\Redo_j$ request to the other, in which case the parties simply repeat the exact same run albeit with fresh randomness. To prevent the Client from post-selecting on the measurement results returned by the Server, the $\Redo_j$ request is allowed only so long as the party asking for it is still supposed to be manipulating the qubits of run $j$.

We show in the next section that this does not impact the blindness nor verifiability of the scheme. This means that a dishonest Server cannot use $\Redo$ requests to trick the Client into accepting an incorrect result.

However, this capability of our protocol has an important practical impact: without this possibility, honest failures of the experimental devices happening during a test run would be counted as a failed test run, thus decreasing drastically the likelihood of successfully performing a verified computation, i.e.\ getting an $\ok$. Here, as they can be safely ignored, the only consequence of experimental failures is to increase the expected number of repeated runs.

Note that, it is possible for the Client to choose the secret variables for all runs in advance, taking into account the probability that a given run $j$ will not be redone due to experimental defects. If this probability of not resetting a run on honest devices is given by $p_{\mathit{succ}}$, then the number of runs to be pre-sampled by the Client is $N = \order{\frac{n}{p_{succ}}}$, with the same proportion of computation $d/n$ and test runs $t/n$.

\paragraph{Exponential Security Amplification Feature.}
The above approach to trap insertion is efficient as the only overhead is the repetition of the same sub-protocol. Yet, as such, it would seemingly only achieve a security bound that is inverse-polynomial in the number of runs. For instance, using a single computation run and $n-1$ test runs would let a $1/n$ chance for the Server to corrupt the computation run. The only previously-known method to reduce the cheating probability to exponentially-low bounds was to insert traps into a single computation run at the expense of using a more complicated graph and then using fault-tolerant quantum error-correction codes to achieve the desired amplification of the security. By restricting our target to classical inputs and outputs, we will prove that it is sufficient to use only a classical repetition error-correcting code to go from inverse polynomial to exponentially-low cheating probability.

There are two practical impacts of this ability. The first one is an economy in terms of qubits and gates required to perform a computation up to the point where fully fault-tolerant Servers are not the only option to implement VBQC schemes. The second relates to security: even though our protocol will be proven secure in case of sequential and parallel executions, each use of the protocol still offers more opportunities for an attacker to succeed. Hence, through exponential security amplification, Clients can rely on a verifiable delegated quantum computation service provider for an indefinite amount of time: Lowering the single run cheating probability to cope with a linear number of potential uses requires only a logarithmic increase of resources devoted to the security of their computations (in terms of number of repetitions).

Note that such amplification is a common and rather intuitive feature for purely classical scenarios where attacks can be correlated across various rounds. Although this claim has been made as well in the quantum case in previous works \cite{fk,KW15,KD19}, it remained up to now unproven. The difficulty that the following section addresses is that, in the quantum realm, attacks can be entangled across rounds in a way that is much more powerful than what is possible with classical correlations.

%% file: sec_files/061_sec_res.tex
We show in this section that the protocol presented above is secure in the Abstract Cryptography Framework of~\cite{MR11}.
In this framework, security means that an Ideal Resource, secure by definition, cannot be distinguished from its real-world implementation, i.e.~the protocol.
To assess this property, all parties accessing the ideal and real setups will be represented by a single distinguisher.
Its purpose is to choose the inputs to both setups and to try to tell them apart by analysing exchanged messages as well as obtained outputs.
When the protocol is meant to be secure against some malicious parties, the distinguisher will additionally have the possibility to deviate from instructions sent to these potentially malicious parties, thus granting it even more opportunities to tell both setups apart (see~\cite{PR14} for an introduction to Abstract Cryptography applied to QKD).

Abstract Cryptography has two main virtues.
First, because the Ideal Resource is secure by design and indistinguishable from the real-world implementation, so is the latter.
This implies a higher standard of security than in other approaches (see e.g.~\cite{KRBM07a} and Section 5.1 of~\cite{PR14}).
Second, it is by construction composable. This means that security holds in situations where the protocol is repeated sequentially, used in parallel or used in conjunction with other protocols.
As a consequence, the security of our protocol will hold in a wide range of situations of practical interest such as when different runs are distributed to different machines to reduce the overall execution time.

Here, instead of following the direct approach to proving security outlined above, we will take a slightly different path.
We will use the results of~\cite{DFPR14} that reduce security of a Verifiable Delegated Quantum Computation Protocol to the conjunction of four stand-alone criteria:
\begin{itemize}
\item $\epsilon_{\mathit{cor}}$-local-correctness, which is satisfied if the protocol with honest players outputs the expected output;
\item $\epsilon_{\mathit{bl}}$-local-blindness, meaning that the malicious Server's state at the end of the protocol is indistinguishable from the one which it could have generated on its own;
\item $\epsilon_{\mathit{ver}}$-local-verifiability, if either the Client accepts a correct computation or aborts at the end of the protocol.
\item $\epsilon_{\mathit{ind}}$-independent-verification, i.e.\ the Server can determine on its own, using the transcript of the protocol and its internal registers, whether the Client will decide to abort or not. 
\end{itemize}
Then, the Local-Reduction Theorem (Corollary 6.9 from~\cite{DFPR14}) gives the following:
If a protocol implements a unitary transformation on classical inputs and is $\epsilon_{\mathit{cor}}$-locally-correct, $\epsilon_{\mathit{bl}}$-locally-blind and $\epsilon_{\mathit{ver}}$-locally-verifiable with $\epsilon_{\mathit{ind}}$-independent verification, then it is $\epsilon$-composably-secure with:
\begin{equation}\label{eq:reduction}
\epsilon = \mathit{max}\qty{\delta, \epsilon_{\mathit{cor}}} \text{ and } \delta := 4\sqrt{2\epsilon_{\mathit{ver}}} + 2\epsilon_{\mathit{bl}} + 2\epsilon_{\mathit{ind}}.
\end{equation}
The reader is reported to Appendix \ref{app:form_def} for formal definitions of the Ideal Resource and the local criteria described above.

With this at hand, we can state our main result:
\begin{theorem}[Security of Protocol \ref{prot:MQ-VBQC}]
\label{theo:sec_VDQC}
For $n = d+t$ such that $d/n$ and $t/n$ are fixed in $(0,1)$ and $w$ such that $w/t$ is fixed in $(0, 1/2k)$, Protocol~\ref{prot:MQ-VBQC} with $d$ computation runs, $t$ test runs, and a maximum number of tolerated failed test runs of $w$ is $\epsilon$-composably-secure with $\epsilon = 4\sqrt{2\epsilon_{\mathit{ver}}}$ and with $\epsilon_{\mathit{ver}}$ exponentially small in $n$.
\end{theorem}

In the paragraphs below, we show that our protocol satisfies each of the stand-alone criteria before combining them to get composable security.

\paragraph{Perfect Local-Correctness.} 
On perfect (non-noisy) devices, local-correctness is implied by the correctness of the underlying UBQC Protocol. 
This is because all the completed computation runs correspond to the same deterministic UBQC computation, and that on such devices, general UBQC Protocols have been proven to be perfectly correct \cite{bfk, DFPR14}.
Thus $\epsilon_{\mathit{cor}} = 0$. 

\paragraph{Perfect Local-Blindness.} 
In case the computation is accepted, each run looks exactly like a UBQC computation to the Server.
Therefore the blindness comes directly from the composability of the various UBQC runs that make our protocol~\cite{DFPR14}.
In case the computation is aborted, we need to take into account the fact that a possibly malicious Server could deduce the position of a trap qubit.
That could be the case if it attacked a single position in the test runs and got caught.
Yet, as the position of the traps is not correlated to the input nor to the computation itself, knowing it does not grant additional attack capabilities to the Server, and blindness is recovered again as a consequence of the blindness of UBQC. 
More detailed statements can be found in Appendix~\ref{app:proof_lb}, where it is also shown that $\Redo$ requests have no effect on the local-blindness of the scheme.

\paragraph{Perfect Local-Independent-Verification.} 
Because in our protocol, the Client shares with the Server whether the computation was a success or an abort, this is trivially verified.

\paragraph{Exponential Local-Verifiability.} 
Local-verifiabili\-ty is satisfied if any deviation by the possibly malicious Server yields a state that is $\epsilon_{\mathit{ver}}$-close to a mixture of the correct output and the $\Abort$ message.
Equivalently, the probability that the Server makes the Client accept an incorrect outcome is bounded by $\epsilon_{\mathit{ver}}$.
Let $d/n$, $t/n$ and $w/t$ be the ratios of test, computation and tolerated failed test runs.
Our protocol's local-verifiability is given by Lemma \ref{thm:verif} and is proven in Appendix \ref{app:proof_ver}.
We give below a sketch of the main ideas yielding the result.

\begin{lemma}[Local Verifiability of Protocol \ref{prot:MQ-VBQC}]
\label{thm:verif}
\label{thm:soundness}
Let $0 < w/t < 1/2k$ and $0 < d/n < 1$ be fixed ratios, for a $k$ the number of different test runs.
Then, Protocol~\ref{prot:MQ-VBQC} is $\epsilon_{\mathit{ver}}$-locally-verifiable for exponentially-low $\epsilon_{\mathit{ver}}$.
\end{lemma}

\begin{proof}[Proof Sketch]
The first step is to describe all the messages received by the Client during the execution of the protocol without making assumptions on the behaviour of the Server.
This comprises the outcomes of the computational runs, but also the measurement of trap qubits and of any other qubit used in the computation or in the tests.
Following~\cite{fk}, this can be expressed as the state one would obtain in the perfect protocol followed by a pure deviation on this state.

The second step consists of using this state to bound the probability of failure, i.e.~the probability of accepting the computation but having the wrong result.
This happens if at least $d/2$ outcomes of the computation runs have had at least one bit-flip, and no more than $w$ test runs have failed.

In the third step, we use the randomisation over the prepared qubits and measurement angles to twirl the deviation of the Server and to reduce it to diagonal form in the Pauli basis.
This further simplifies the expression for the bound.

The fourth step exploits this reduced form of the attack by noticing that it can be dealt with in a classical fashion.
To this end, possible attacks are classified using two criteria:
\begin{enumerate}
	\item Does the attack affect at least $d/2$ computation runs? Only such attacks stand a chance to corrupt the result of the computation, otherwise the repetition code automatically corrects the deviations.
	\item Does the attack make less than $w$ of the $t$ test runs fail? Only then will the deviations be tolerated without triggering a client-side abort.
\end{enumerate}
Depending on the answer to the questions above, the attack falls in one of four regimes.
Optimally, the protocol would abort if and only if the result of the computation is corrupted.
Clearly, we cannot hope to perfectly achieve this.
We therefore must take into account two types of incorrect categorisation.
A \emph{false positive} happens when the protocol aborts although less than $d/2$ computation runs have been affected.
While this is undesirable behaviour in terms of noise-tolerance, it does not affect security.
Since we are here analysing the verifiability of our protocol, we are solely concerned about \emph{false nagatives}: the attack affects at least $d/2$ computation runs and no abort is triggered.
To achieve a satisfying level of security, no attack should falls into this regime with more than negligible probability.

For intuition's sake, we give here an analysis of the average case (Figure~\ref{fig:intuition}).
If the Server deviates in exactly half of the runs on a single qubit (which we suppose to be sufficient for corrupting the computation), the number of affected computation runs will be $d/2$ on average.
In other words, there are good chances that just enough computation runs are affected to corrupt the final result.
This is because the attacker is blind and hence the deviations are randomly distributed over computation and test runs.
Similarly, we expect the number of affected test runs to be $t/2$.
Considering that any qubit in a run has a probability of $1/k$ to be a trap, we expect the number of failed test runs, i.e.~the number of test runs with at least one affected trap, to be $t/2k$.
As a consequence, choosing $w \geq t/2k$ cannot lead to a secure protocol, since the simple attack described above has a non-negligible probability of corrupting the final result while remaining unnoticed.
Conversely, setting $w \leq t/2k$ foils this strategy.

The proof presented in Appendix~\ref{app:proof_ver} goes beyond this average case analysis by showing that this attack is essentially optimal.
It uses concentration bounds for the underlying probability distributions to obtain precise bounds that are exponentially-low in the various parameters.
\end{proof}

\begin{figure*}[ht]
  	\centering
	\includegraphics[width=0.8\textwidth]{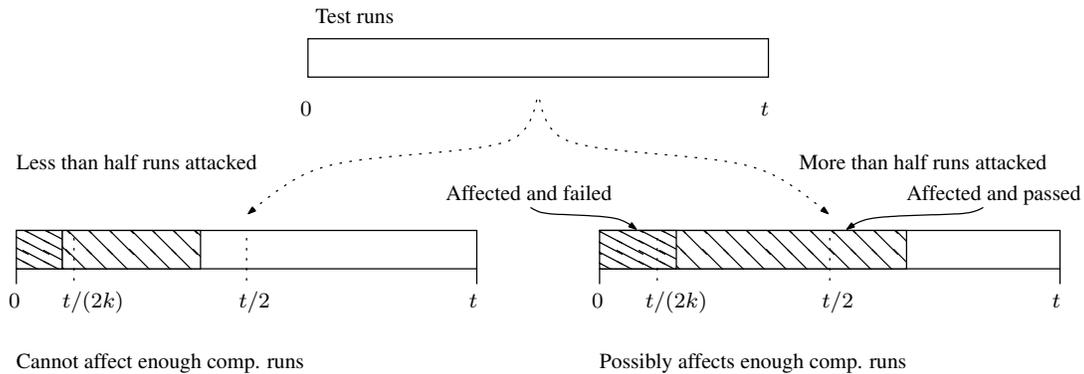}
	\caption{Because of blindness, an attack on less than half of the runs is likely to affect less than $d/2$ computation runs.
Such attack being error corrected, the protocol should output a result.
On the contrary, if an attack is performed on more than half the runs, it has a chance to corrupt the computation and the result should not be trusted.
Likewise, in the first case, less than $t/2$ test runs should be affected, while in the second more than $t/2$ should be.
Yet, there are only $1/k$ test qubits per test run meaning that affected test runs turn into failed test runs with $1/k$ probability.
If the threshold value $w$ is set below $t/2k$, aborts are thrown only when there is a effective risk of tampering.
Here, we arbitrarily test $k=4$.}
  	\label{fig:intuition}
\end{figure*}

\paragraph{Proof of Exponential Composable-Security.} 
Our protocol has perfect correctness (for noiseless devices), blindness and input-independent verification. In addition, it is $\epsilon_{\mathit{ver}}$-locally-verifiable with $\epsilon_{\mathit{ver}}$ exponentially small in $n$. Therefore, by the Local-Reduction Theorem, it is $\epsilon$-composably-secure with $\epsilon = \delta = 4\sqrt{2\epsilon_{\mathit{ver}}}$ and $\epsilon $ exponentially small in $n$.
Note that because we used the Local-Reduction Theorem to obtain composable security, we incurred an additional square root on our verifiability bound given by Equation~\ref{eq:reduction} and needed to satisfy the additional independence property.
This is of course not required if the protocol is meant to be used in a stand-alone setting such as in early NISQ-era experiments.

%% file: sec_files/07bis_rob.tex
\subsubsection{Local-Correctness on Honest-but-Noisy Devices}

The local-correctness property discussed in the previous section did not take into account device imperfections. In fact, the analysis of blindness and verification makes no distinction between these imperfections and potentially malicious behaviours. 
Although satisfying these properties makes our protocol a concrete implementation of the Ideal Resource for Verifiable Delegated Quantum Computation, it could still fall short of expectations in terms of usability. 
Fortunately, for a class of realistic imperfections, our protocol has the additional property of being capable of correcting their impact and accepting with high probability. The final outcome is then the same as that obtained on noiseless devices with honest participants. 

This additional \emph{noise-robustness} property, the main innovation of this paper, amounts to prove that Protocol~\ref{prot:MQ-VBQC} satisfies the local-correctness property with negligible $\epsilon_{\mathit{cor}}$ for noisy honest Client and/or Server devices. We prove this property under the following restrictions:

\begin{itemize}
\item The noise can be modelled by run-dependent Markovian processes -- i.e.\ a possibly different arbitrary CPTP map acting on each run.
\item The probability that at least one of the trap measurements fails in any single test run is upper-bounded by some constant $\pmax < 1/2$ and lower-bounded by $\pmin \leq \pmax$.
\end{itemize}

Theorem~\ref{thm:correctness} states that, in order for the protocol to terminate correctly with overwhelming probability on these noisy devices, $w$ should be chosen such that $w/t > \pmax$. Conversely, for any choice of $w/t < \pmin$, we show that the protocol aborts with overwhelming probability. See Appendix \ref{app:proof_rob} for its formal version and proof.

\begin{theorem}[Local-Correctness of VDQC Protocol on Noisy Devices (Informal)]
\label{thm:correctness}
Assume a Markovian run-dependent model for the noise on Client and Server devices and let $\pmin \leq \pmax < 1/2$ be respectively a lower and an upper-bound on the probability that at least one of the trap measurement outcomes in a single test run is incorrect. If $w/t > \pmax$, Protocol~\ref{prot:MQ-VBQC} is $\epsilon_{\mathit{cor}}$-locally-correct with exponentially low $\epsilon_{\mathit{cor}}$. On the other hand, if $w/t < \pmin$, then the probability that Protocol~\ref{prot:MQ-VBQC} terminates without aborting is exponentially low.

\end{theorem}

\begin{proof}[Proof Sketch]

Similarly to the previous proof for verifiability, we will here give an intuitive average-case analysis. The full proof uses similar concentration bounds. 

Starting with the case where $w > t\pmax$, we need to analyse the probability that the computation and the Client accepts it without aborting. This happens if the noise has not disturbed too many computation and test runs (less than $d/2$ and $w$ respectively). These events are independent and are analysed as such in the following since the noise is independent across runs and the nature of each run is chosen uniformly at random. 

If we assume that any corruption on a single qubit of a computation run leads to a corrupted computation, we deduce that $\pmax$ is also an upper-bound on the probability that the outcome of the computation is incorrect. We can then upper-bound the number of failed computation runs by a random variable that follows $(d, \pmax)$-binomial distribution, of mean $d\pmax$. Since $\pmax < 1/2$, we have that $d\pmax < d/2$ and therefore the noise will not corrupt the computation.

Regarding the test runs, a similar reasoning gives that it is also possible to upper-bound the number of failed test runs by random variable following a $(t, \pmax)$-binomial distribution, of mean $t\pmax < w$. We conclude that not enough test runs are corrupted to trigger an abort and the Client accepts the (correct) outcome of the computation.

We now focus on the case where $w < t\pmin$ and show that the Client rejects the computation. The number of failed test runs is lower-bounded by a random variable following a $(t, \pmin)$-binomial distribution, of mean $t\pmin$. This number of failed tests is higher than the acceptable threshold and the Client therefore aborts. This is a case of \emph{false-positive} since the outcome is still correct in this scenario.

\end{proof}

Since the results for blindness, blindness, input-independent verification and verifiability already integrate the noise in their analyses (as explained below, they consider the most general deviation, which includes noise), this new bound concerning local-correctness on noisy devices can also be used alongside these previous bounds using the Local-Reduction Theorem from~\cite{DFPR14}, yielding in this case a value for $\epsilon = \mathit{max}\qty{\delta, \epsilon_{\mathit{cor}}}$ that may now depend on the noise level of the devices since $\epsilon_{\mathit{cor}} > 0$ (however, note that if $\delta > \epsilon_{\mathit{cor}}$, then the noise has no impact on the total $\epsilon$).

\subsubsection{Role of Noise Assumptions}\label{sec:noise_assump}

As stated in the last section, our security proof does not rely on any assumption regarding a specific form or strength of the noise. On the contrary, it simply considers any deviation as potentially malicious and showed that the protocol would provide information theoretic verification and blindness of the computation. 

The assumptions introduced in this section serve a different purpose: when the imperfections of the devices are light enough, we show that these will systematically correspond to attacks whose impact on the computation is error-correctable. The noise models include independent Pauli operators acting on qubits, but also more general operators that are independent between various runs, etc. As a consequence, under such mild restrictions, the computation will accept with high probability. The Client will hence not only get a security guarantee but also a performance guarantee in that it will obtain a result.

Consequently, there is no security risk in first probing the device to find out about the noise level and then using it to set the admissible ratio of failed test runs $w/t$ to a compatible value. The value for $w/t$ might even be adjusted between two executions of the full protocol to cope with drifting values of noise. An over-inflated value of $w/t$ only results in superfluous repetitions and hence a longer running time for the protocol. Conversely, setting the value of $w/t$ too low carries the risk of aborting most of the time and thus not being able to ever complete the computation, as is the case with previous protocols without error-correction capabilities.

Regarding the assumption that the noise maps between each run are independent, this can realistically be achieved in an experimental setup by simply waiting long enough between each run for all the states to decohere. This guarantees that no information about a previous run can seep in the next execution. This of course prolongs the duration of the experiment, but considering the low coherence time of quantum memories in the NISQ era (which is precisely the regime that this protocol targets), this overhead is not too prohibitive.

\subsubsection{Link to Fault-Tolerance}

We want to emphasize that our scheme does \emph{not} rely on nor provide fault-tolerance. It is quite the opposite in that it decouples the error-correction scheme devoted to fault-tolerant computing from the one that ensures robust verification. More precisely, we define the perceived noise level as the average ratio of test runs that fail given a specific noise model. Given the importance of this parameter, it is legitimate to wonder how it can be optimized.

It is a global metric in the sense that it is obviously easier to attain a perceived noise level below our \(1/4\) threshold for a $2$-qubit computation\footnote{For $2$-qubit computation, there are necessarily $2$ types of test runs.} than it is for a $100$-qubit computation. This can be understood in the following way: the built-in error correction capability does not stop errors from propagating into the result, it only tries to recover the correct result from the noisy outcomes of each unprotected computational run. The other way around, a fault-tolerant scheme prevents error-propagation and would be useful in lowering the ratio of failed test runs in large computations and arrive at, for example, the \(1/4\) ratio of failed test runs required to perform the verified computation.

An advantage of this decoupling between fault-tolerant computation and robust verification is that, while fault-tolerance is likely unavoidable for large computations on yet-to-be-created systems, near-term and sufficiently noiseless devices (at the physical level) will allow for verified intermediate-scale computations by using our scheme. This would be unreachable if performed in a fault-tolerant way due to large resource overheads.

%% file: sec_files/08_concl.tex
\subsubsection{Overhead Optimisation}

Once the noise level $\pmax$ has been determined and used to constrain $w$ as explained in Section \ref{sec:noise_rob}, we can look at ways to optimize the resource overhead for robust verification in terms of excess number of runs compared to standard MBQC\footnote{Note that this is the only source of overhead as each run requires the same resources as the original MBQC.}.

The first obvious parameter influencing the overhead is the number $k$ of different types of test runs. Higher values of \(k\) induce lower amounts of tolerated noise and a higher number of repetitions. In our scheme, $k$ also corresponds to the number of colours in the graph colouring chosen by the Client on the Server's graph $G$. While the problem of finding an optimal (i.e.\ minimal) graph colouring is $\NPC$ for general graphs, there exist efficient algorithms to compute approximately optimal graph colourings.

For any general graph, a greedy algorithm yields a $k$-colouring for $k \leq D(G) + 1$, with $D(G)$ being the maximum degree of $G$. Also note that most graphs used in MBQC are planar and the celebrated $4$-Colour Theorem states that any planar graph is $4$-colourable, hence bounding \(k\) by $4$. Efficient algorithms to find such a colouring exist (quadratic in the number of vertices of the graph~\cite{RSST96}) which would then be of practical interest in designing robust verifiable schemes. Furthermore, as \(k= 2\) is the best possible value for our scheme and because it is efficient to check if a graph is bi-partite, the value \(k=2\) should be tested. Note that the brickwork graph (which is universal for quantum computations) and all dotted (edge-decorated) graphs are bipartite. In contrast, testing the case of a $3$-colourable graph is $\NPC$, so for large planar graphs the Client may have to choose a $4$-colouring instead at the expense of more repetitions of the protocol for attaining the same verification bound.

Once \(k\) is fixed, for targeted values of perceived noise level resistance, acceptance and failure probabilities, numerical optimizations can be used to determine the best values for the total number of runs \(n\), the ratio of test runs \(t/n\), and the ratio of test runs allowed to fail \(w/t\) using equation~\ref{eq:bound}.

\subsubsection{Link to Certification and Benchmarking}

Finally, we want to point out a connection between certification and verification that stems directly from the presented protocol. As mentioned in the introduction, two broad strategies can be pursued when one wants to give guarantees with respect to a device sold to clients.

On one hand, the provider and the device can be certified by a trusted third party and the provider commits to manufacture the device that has been certified. The commitment is enforced not by design but legally. This is often chosen for efficiency reasons: the certification is done once and in case of widespread services or devices, the cost (in time, money, effort) to certify it is absorbed into the volume of service or devices provided. Another reason for choosing this form of certification is that it offers a natural way to cope with imperfect devices. Most devices are certified to within some acceptable range of performance describing its nominal behaviour.

This however has caveats, as recent years have proven. For the commitment to be effective, there needs to be a reasonable chance to catch deviations from the certified behaviour which supposes in turn that devices are prevented from sensing whether they are being tested or not.\footnote{This was exactly the strategy developed by some car makers: by guessing when the engine was being run on a test-bench they would reduce its power as to pass the tests while offering a widely different behaviour when in real conditions. It was only after independent associations measured the emissions in road-like conditions and found a gap large enough that it could not be explained by variations in physical conditions that extensive search was conducted and the deviation discovered. Had the gap been smaller, it would have most certainly gone undetected.}

On the other hand, the provider can chose to opt for a verifiable device (or service). In that case the security is better as there is no commitment required, only fact-based trust dependent on a series of tests. The trouble with such strategy is - or rather was up to now - its high overhead in the context of quantum computing, making it inaccessible and extremely expensive in any foreseeable future.

Our results show that for certain classes of computations this does not need to be the case. In fact, the best of both world can be combined. Test runs are indeed probing whether our device and computation is abiding by some certification standards expressed in terms of an effective noise level. This is done continuously through the computation in order to prevent the device from adapting its behaviour. Since the computation is blind, even a fully malicious adversary cannot successfully fool the protocol. So in effect, blindness allows to combine computation and certification to arrive at verification, retaining the efficiency and imperfection tolerance of certification while keeping the unconditional security of verified schemes. While this is clearly more expensive than simple certification, this overhead should be acceptable for a wide range of practical situations. A natural open question, that we leave for further research, is whether this strategy can be extended to other protocols and if there are situations where other schemes are more efficient.

%% file: sec_files/app_tail_bounds.tex
The following definitions and lemmata are useful tools for our proof.

\begin{definition}[Hypergeometric distribution \cite{wiki:hypergeometric}]
        Let $N, K, n \in \mathbb{N}$ with $0 \leq n,K \leq N$.  A random variable $X$ is said to follow the \emph{hypergeometric distribution}, denoted as $X\sim\operatorname{Hypergeometric}(N,K,n)$, if its probability mass function is described by
	\begin{align*}
		\Pr \left[ X = k \right] = \frac{\binom{K}{k}\binom{N-K}{n-k}}{\binom{N}{n}}.
	\end{align*}
	As one possible interpretation, $X$ describes the number of drawn marked items when drawing $n$ items from a set of size $N$ containing $K$ marked items, \emph{without replacement}.
\end{definition}

\begin{lemma}[Tail bound for the hypergeometric distribution \cite{wiki:hypergeometric}]
	Let $X\sim\operatorname{Hypergeometric}(N,K,n)$ be a random variable and $0 < t < K/N$. It then holds that
	\begin{align*}
	  &\Pr \left[ X \leq \left( \frac{K}{N}-t \right) n \right] \leq \exp \left( -2t^2n \right).
	\end{align*}
\end{lemma}

\begin{corollary} \label{cor:lower_tail}
	Let $X\sim\operatorname{Hypergeometric}(N,K,n)$ be a random variable and $0 < \lambda < \frac{nK}{N}$. It then holds that
	\begin{align*}
		\Pr \left[ X \leq \lambda \right] \leq \exp \left( -2n \left( \frac{K}{N} - \frac{\lambda}{n} \right)^2 \right).
	\end{align*}
\end{corollary}

\begin{lemma}[Serfling's bound for the hypergeometric distribution \cite{Greene_2017,Serf74}]
	Let $X\sim\operatorname{Hypergeometric}(N,K,n)$ be a random variable and $\lambda > 0$. It then holds that
	\begin{align*}
		\Pr \left[ \sqrt{n} \left( \frac{X}{n} - \frac{N}{K} \right) \geq \lambda \right] \leq \exp \left( - \frac{2\lambda^2}{1-\frac{n-1}{N}} \right).
	\end{align*}
\end{lemma}

\begin{corollary} \label{cor:upper_tail}
	Let $X\sim\operatorname{Hypergeometric}(N,K,n)$ be a random variable and $\lambda > \frac{nK}{N}$. It then holds that
	\begin{align*}
		\Pr \left[ X \geq \lambda \right] \leq \exp \left( -2n \left( \frac{\lambda}{n} - \frac{K}{N} \right)^2 \right).
	\end{align*}
\end{corollary}

Note the symmetry of Corollary~\ref{cor:lower_tail} and Corollary~\ref{cor:upper_tail}.

\begin{lemma}[Hoeffding's inequality for the binomial distribution] \label{lemma:hoeffding_binomial}
	Let $X\sim\operatorname{Binomial}(n,p)$ be a random variable. For any $k \leq np$ it then holds that
	\begin{align*}
		\Pr \left[ X \leq k \right] \leq \exp \left( -2\frac{(np-k)^2}{n} \right).
	\end{align*}
	Similarly, for any $k \geq np$ it holds that
	\begin{align*}
		\Pr \left[ X \geq k \right] \leq \exp \left( -2\frac{(np-k)^2}{n} \right).
	\end{align*}
\end{lemma}

%% file: sec_files/app_sec_def.tex
We model $N$-round two party protocols between players $A$ (the honest Client) and $B$ (the potentially dishonest Server) as a succession of $2N$-CPTP maps $\{\mathcal{E}_i\}_{i\in [1,N]}$ and $\{\mathcal{F}_j\}_{j\in[1,N]}$. The maps $\{\mathcal E_i\}_i$ act on $\mathcal A$, $A$'s register, and $\mathcal C$, a shared communication register between $A$ and $B$. Similarly, the maps $\mathcal \{F_j\}_j$ act on $\mathcal B$ and $\mathcal C$. Note that $\mathcal B$ and the maps $\{\mathcal F_j\}_j$ can be chosen arbitrarily by $B$ and thus, unless $B$ is specified to be behaving honestly, there is no guarantee that they are those implied by our protocol. Since we are only interested in protocols where $A$ is providing a classical input $x$, we will equivalently write the input as the corresponding computational basis state $\ket x$ used to initialize $\mathcal A$, whereas $\mathcal B$ and $\mathcal C$ are initialized in a fixed state $\ket 0$.

Below, we denote by $\Delta(\rho, \sigma) = \frac{1}{2} \|\rho-\sigma\|$, the distance on the set of density matrices induced by the trace norm $\| \rho \| = \Tr \sqrt{\rho^\dagger\rho}$. We first define $\mathcal S$ the ideal resource for verifiable delegated quantum computation and then the local-properties from \cite{DFPR14}.

\paragraph{Ideal Resource for Verifiable Delegated Quantum Computation.} 
The ideal resource $\mathcal S$ has interfaces for two parties, $A$ and $B$. The $A$-interface takes two inputs: a classical input string $x$ and the description of $\mathcal U$, the computation to perform. The $B$-interface is filtered by a bit $b$. When $b=0$, there is no further legitimate input from $B$, while for $b=1$, it is allowed to send a bit $c$ that determines the output of the computation available at $A$'s interface. When $b=0$ or $c=0$, the output at $A$'s interface is equal to $\mathcal{M}_{\mathit{Comp}} \circ \mathcal{U}(\ket x)$, where $\mathcal{M}_{\mathit{Comp}}$ is the computational basis measurement. This corresponds to a ``no cheating'' behaviour. When $c=1$, $B$ decided to cheat and $A$ receives the $\Abort$ message which can be given as a quantum state of $\mathcal A$ which is taken orthogonal to any other possible output state. At $B$'s interface, $\mathcal S$ outputs nothing for $b=0$ while for $b=1$, $B$ receives $l(\mathcal U, x)$, the permitted leakage. For generic MBQC computations, the permitted leakage is set to $G$, the graph used in the computation. When $G$ is a universal graph for MBQC computation, the permitted leakage reduces to an upper-bound on the size of the computation $\# \mathcal U$.

For this ideal resource, the blindness is an immediate consequence of the server receiving at most the permitted leak, while verifiability is a consequence of the computation being correct when the server is not cheating while being aborted otherwise.

\paragraph{$\epsilon_{\mathit{cor}}$-Local-Correctness.} 
Let $\mathcal{P}_{AB}$ be a two-party protocol as defined above with the honest CPTP maps for players A and B. We say that such a protocol implementing $\mathcal{U}$ is \emph{$\epsilon_{cor}$-locally-correct} if for all possible inputs $x$ for $A$ we have:

\begin{equation}
\label{eq:correct}
\Delta\left(\Tr_B \circ \mathcal{P}_{AB} (\ket x), \mathcal U (\ket x)\right) \leq \epsilon_{cor}
\end{equation}

\paragraph{$\epsilon_{\mathit{bl}}$-Local-Blindness.}
\label{proof:local_blindness}
Let $\mathcal{P}_{AB}$ be a two-party protocol as defined above, and where the maps $\{\mathcal E_i\}_i$ are the honest maps. We say that such protocol is \emph{$\epsilon_{\mathit{bl}}$-locally-blind} if, for each choice of $\{\mathcal F_i\}_i$ there exists a CPTP map $\mathcal F' : L(\mathcal B) \rightarrow L(\mathcal B)$ such that, for all inputs $x$ for $A$, we have:

\begin{equation}
\label{eq:blind}
\Delta\left(\Tr_A \circ \mathcal{P}_{AB} (\rho), \mathcal F' \circ \Tr_A(\ket x)\right) \leq \epsilon_{\mathit{bl}}
\end{equation}

\paragraph{$\epsilon_{\mathit{ind}}$-Independent Verification.} 
Let $\mathcal{P}_{AB}$ be a verifiable 2-party protocol as defined above, where the maps $\{\mathcal E_i\}_i$ are the honest maps. Let $\bar B$ be a qubit extending $B$'s register and initialized in $\ket 0$. Let $\mathcal{Q}_{A\bar{B}} : L(\mathcal{A}\otimes\bar{\mathcal{B}}) \rightarrow L(\mathcal{A}\otimes\bar{\mathcal{B}})$ be a CPTP map which, conditioned on $\mathcal{A}$ containing the state $\ket{\Abort}$, switches the state in $\bar{\mathcal{B}}$ from $\ket 0$ to $\ket 1$ and does nothing in the other cases. 

We say that such a protocol's verification procedure is \emph{$\epsilon_{\mathit{ind}}$-independent} from player A's input if there exists CPTP maps $\mathcal{F}'_i : L(\mathcal{C} \otimes \mathcal{B}\otimes\bar{\mathcal{B}}) \rightarrow L(\mathcal{C} \otimes \mathcal{B}\otimes\bar{\mathcal{B}})$ such that:

\begin{equation}
\label{eq:indep}
\Delta\left(\Tr_A \circ \mathcal{Q}_{A\bar{B}} \circ \mathcal{P}_{AB}(\rho), \Tr_A \circ \mathcal{P}'_{AB\bar{B}}(\rho)\right) \leq \epsilon_{\mathit{ind}}
\end{equation}

where

\[
\mathcal{P}'_{AB\bar{B}} := \mathcal{E}_1 \circ \mathcal{F}'_1 \circ \ldots \circ \mathcal{E}_n \circ \mathcal{F}'_n
\]

\paragraph{$\epsilon_{\mathit{ver}}$-Local-Verifiability.} 
Let $\mathcal{P}_{AB}$ be 2-party protocols as defined above where the maps for $A$ are the honest maps, while the maps $\{\mathcal F_j\}_j$ for $B$ are not necessarily corresponding to the ideal (honest) ones. Let $x$ be the input given by $A$ in the form of a computational state $\ket x$ and $\mathcal U$ the computation it wants to perform. The protocols $\mathcal P_{AB}$ are \emph{$\epsilon_{\mathit{ver}}$-locally-verifiable} for $A$ if for each choice of CPTP maps $\{\mathcal F_j\}_{j}$, there exists $p \in [0,1]$ such that we have:

\[
\Delta\Bigl(\tr_B \mathcal P_{AB}(\ket x), p \mathcal U(\ket x) + (1-p) \ketbra{\Abort}) \leq \epsilon_{\mathit{ver}}
\]

%% file: sec_files/069_local_blindness.tex
\begin{proof}
To prove that Equation~\ref{eq:blind} holds for $\epsilon_{\mathit{bl}} = 0$, first note that at the end of our protocol, the Client $A$ reveals to the Server $B$ whether the computation was accepted or  aborted.
Hence, each case can be analyzed separately.
Second, we show that the interrupted runs that have triggered a $\Redo$ can be safely ignored.
Indeed, each one of them is the begining of an interrupted UBQC computation, and, because UBQC is composable and perfectly blind \cite{DFPR14}, no information can leak to the Server through the transmitted qubits.
In addition, our protocol restricts the honest party $A$ in its ability to emit $\Redo$ requests, so that no correlations are created between the index of the interrupted runs and $\mathcal U$ or the secret random parameters used in the runs (angle and measurement padding, and trap preparations).
As a consequence, from the point of view of $B$, the state of the interrupted runs is completely independent of the state of the non-interrupted ones and does not contain information regarding the input, computation or secret parameters.
That is, its partial trace over $A$ can be generated by $B$ alone.

For the non-interrupted runs, we can invoke the same kind independence argument between the computation runs and the test runs.
As a result blindness of our protocol stems from the blindness of the underlying computation runs.
In case the full protocol is a success, we can rely on the composability of the perfect blindness of each UBQC computation run to have perfect local-blindness.
For an abort, we can consider a situation that is more advantageous for $B$ by supposing that alongside the $\Abort$ message sent by $A$, it also gives away the location of the trap qubits.
In this modified situation, the knowledge of the computation being aborted does not bring additional information to $B$ as it only reveals that one of the attacked position was a trap qubit, which $B$ now already knows.
Using our independence argument between trap location on the one hand and the inputs, computation and other secret parameters, we conclude that revealing the location of the trap qubits does not affect the blindness of the computation runs.
Hence, using composability again and combining the abort and accept cases, we arrive at Equation~\ref{eq:blind} with $\epsilon_{\mathit{bl}} = 0$.

\end{proof}

%% file: sec_files/07_verif.tex
\begin{proof}[Proof of Lemma~\ref{thm:verif}]

  We will take a direct approach for proving Lemma~\ref{thm:verif} by bounding the probability of yielding a wrong output while not aborting. To do so, we consider the state of the {\em combined computation} as if it was a single verified computation and not made of separate sequential runs. Once again, in our protocol, because the parties can only ask for redoing a run independently of the input, computation, used randomness and of the output of the computation itself (comprising the result of trap measurements), interrupted runs can be safely ignored in the verification analysis as the state corresponding to these runs is uncorrelated to that of the completed runs. The combined computation view will be useful as we want to consider the Server $B$ performing any kind of attack. In particular, it could decide to perform some action on a qubit given measurements in one or several of the underlying runs, or to entangle the various underlying runs together. Yet, for each qubit of the combined computation, we will continue to refer to the underlying run this qubit would belong to if the computation was done using sequential runs.

\paragraph{Output of the combined computation.}
We first consider the density operator $B(\{\mathcal F_j\}_j, \nu)$ that corresponds to the classical messages the Client $A$ receives during its interaction with $B$, comprising the final message containing the encrypted measurement outcomes. Below, the CPTP maps $\{\mathcal F_j\}_j$ represent the chosen behavior of $B$ and possibly act on the combined computation as a whole, and not only run by run. By representing the classical messages as quantum states in the computational basis, we can always write:
\begin{align}
\label{eq:output}
& B(\{\mathcal F_j\}_j,\nu) = \Tr_B \bigg\{ \sum_b \ketbra{b+c_r}{b} \mathcal F \mathcal{P} \nonumber \\
& \quad \left(\ketbra{0}_B \otimes \ketbra{\Psi^{\nu,b}}\right) \mathcal{P}^\dagger {\mathcal F}^\dagger  \ketbra{b}{b+c_r}\bigg\}
\end{align}
where \(b\) is the list of measurement outcomes defining the computation branch; \(\nu\) is a composite index relative to the secret parameters chosen by $A$, i.e.~the type of each underlying run, the padding of the measurement angles and measurements outcomes and the trap setup; \(\ketbra{b+c_r}{b}\) ensures that only the part corresponding to the current computation branch is taken into account and removes the One-Time-Pad encryption on non-output and non-trap qubits while leaving output and trap qubits unaffected, i.e.~encrypted; \(\ketbra{0}_B\) is some internal register for $B$ in a fixed initial state; and \(\ket{\Psi^{\nu,b}}\) is the state of the qubits sent by $A$ to $B$ at the beginning of the protocol tensored with quantum states representing the measurement angles of the computation branch $b$.

To obtain this result, we can follow the line of proof of~\cite{fk} and~\cite{KW15} applied to the combined computation. This works by noting that for a given computation branch $b$ and given random parameters $\nu$, all the measurement angles are fully determined. Therefore, provided that the computation branch is $b$, we can include the measurement angles into the initial state. This defines $\ket{\Psi^{\nu,b}}$. Then, each $\mathcal F_j$ is decomposed into an honest part and a pure deviation. All the deviations are commuted and collected into $\mathcal F$ applied after $\mathcal P$, the unitary part of honest protocol, is applied. The projections onto $\ket b$ then ensures that after the deviation induced by $B$ the perceived computation branch is $b$. This, together with the decrypting of non-output non-trap qubits, gives Equation~\ref{eq:output}.

\paragraph{Probability of failure.}
A failure for the combined computation occurs when the result after the majority vote is incorrect while the computation is accepted.

The combined computation being deterministic, we can define \(P_\perp\), the projector onto the subspace of incorrect states for the output qubits before the majority vote. Yet, for the combined computation to be accepted, no more than \(w\) test runs have a trap qubit measurement outcome opposite to what was expected. Let \(\fT\) denote the set of trap qubits which is determined by $T$, the set of test runs, and the type of each test run. In absence of any deviation on the combined computation, their expected value is \(\ket{r_\fT} = \bigotimes_{\ft \in \fT}\ket{r_\ft}\) where \(r_{\fT} = (r_{\ft})_{\ft \in \fT}\) denotes the measurement outcome padding values restricted to trap qubits. Therefore, the projector onto the states of the trap qubits yielding to an accepted combined computation can be written as \(\sum_{\fw \in \fW} X_{\fT}^{\fw} \ketbra{r_{\fT}} X_{\fT}^{\fw}\) with \(X_{\fT}^{\fw} = \bigotimes_{\ft\in \fT} X_{\ft}^{\fw_{\ft}}\), and where \(\fW\) is the set of length \(|\fT|\) binary vectors \(\fw\) that have at least a one in no more than \(w\) underlying (test) runs. Combining the projector onto incorrect output and the one for accepted computation, we obtain the probability of failure:

\begin{align*}
& \Pr(\mathrm{fail}) = \\
& \quad \sum_{\nu}\sum_{\fw \in \fW}\sum_{b,k,\sigma,\sigma'} \Pr(\nu) \Tr \Big \{ \left(P_\perp \otimes  X_{\fT}^{\fw} \ketbra{r_{\fT}} X_{\fT}^{\fw} \right) \times \\
& \quad \left( \alpha_{k\sigma}\alpha^*_{k\sigma'} \ketbra{b+c_r}{b} \sigma \mathcal{P}\ketbra{\Psi^{\nu, b}} \mathcal{P}^\dagger \sigma' \ketbra{b}{b+c_r}\right) \Big\}  
\end{align*}
where \(\mathcal F\) has been decomposed into Kraus operators indexed by \(k\), that were in turn decomposed onto the Pauli basis through the coefficients \(\alpha_{k\sigma}\) and \(\alpha_{k\sigma'}\). Consequently, \(\sigma\) and \(\sigma'\) are Pauli matrices.

\paragraph{Necessary condition for failure.}
The difficulty with the above expression for the probability of failure consists in determining the exact form of \(P_\perp\) and manipulating it. Instead, we will derive a coarse necessary condition for the final state of the non-trap qubits to be in the subspace defined by \(P_\perp\). Then, we will upper bound \(\Pr(\mathrm{fail})\) by evaluating the probability of satisfying our necessary condition while accepting the whole computation.

First, note that the output of the computation being classical and deterministic, we can write the correct decrypted output state as \(\ketbra{s_{\fO}}\) for some length \(|\fO|\) binary vector \(s_{\fO}\) over the set of output qubit positions \(\fO\) of the combined computation. Next, as for the trap qubits, the value sent to the Client is One-Time-Padded by the value of the random parameter \(r_{\fO}\) to preserve blindness of the Server (i.e.~\(c_r\) is 0 for output qubits). Hence, the state of the output qubits received by the Client in absence of deviation is \(\ketbra{s_{\fO} + r_{\fO}}\).

Now, because the result of the computation is the majority vote of the measurement outcomes for the output qubit for each underlying computation run, each result bit is protected by a length \(d\) repetition code. All attacks resulting in less than \(d/2\) non-trivially affected underlying computational runs will be corrected. Conversely, for a failure to happen, it is necessary that at least \(d/2\) underlying computation runs are non-trivially affected by the attack \(\Omega\). This means that the subspace stabilized by \(P_\perp\) is also stabilized by the coarser projection operator \(\sum_{\fv\in \fV} X_{\fO}^{\fv} \ketbra{s_{\fO} + r_{\fO}} X_{\fO}^{\fv}\), where \(\fV\) is the set of binary vectors over \(\fO\) with at least one non zero position in at least \(d/2\) underlying (computational) runs. As a consequence, the latter projector can be used to replace \(P_\perp\) which yields an upper bound on \(\Pr(\mathrm{fail})\):

\begin{align*}
& \Pr(\mathrm{fail}) \leq \\
& \quad \sum_{\nu}\sum_{\fv \in \fV, \fw \in \fW}\sum_{b',k,\sigma,\sigma'} \Pr(\nu) \alpha_{k\sigma}\alpha^*_{k\sigma'} \Tr \Bigg\{ \\
& \quad (X_{\fO}^{\fv}\otimes X_{\fT}^{\fw}) \left(\ketbra{s_{\fO}+r_{\fO}} \otimes \ketbra{r_{\fT}}\right) (X_{\fO}^{\fv} \otimes X_{\fT}^{\fw}) \times \\
& \quad \ketbra{b+c_r}{b} \sigma \mathcal{P}\ketbra{\Psi^{\nu, b}} \mathcal{P}^\dagger \sigma' \ketbra{b}{b+c_r} \Bigg\}
\end{align*}
where \(b'\) is the binary vector obtained from \(b\) by restricting it to non-output, non-trap qubits (i.e.~\(b = (b_O, b_T, b')\). The above equation is obtained from the simple equality \(\sum_{b} ( \bra{s_{\fO} + r_{\fO}}\otimes\bra{r_{\fT}}) (X_{\fO}^{\fv}\otimes X_{\fT}^{\fw}) \ketbra{b+c_r}{b} = \sum_{b'} (\ket{s_{\fO} + r_{\fO}}\otimes \ket{r_{\fT}}\otimes \ket{b'}) (X_{\fO}^{\fv} \otimes X_{\fT}^{\fw})\) since \(c_r = 0\) for output and trap qubits, and the circularity of the trace.

\paragraph{Using blindness of the scheme.}
The design of the protocol yielding the combined computation ensures blindness. This implies that the resulting state of any set of qubits after applying \(\mathcal P\) and taking the average over their possible random preparations parameters is a completely mixed state. This can be applied in the above equation for the set of non-output and non-trap qubits. For output and trap qubits, we need first to compute inner products before taking the sum over their random preparation parameters \(\nu_{\fO}\) and \(\nu_{\fT}\) respectively. However, we know that the perfect protocol produces the traps in their expected states \(\ket{s_\fo + r_\fo}\) and \(\ket{r_\ft}\). This gives, using the circularity of the trace:

\begin{align*}
\Pr(\mathrm{fail}) \leq & 
\sum_{\nu_{\fO},\nu_{\fT}}\sum_{\fv \in \fV, \fw \in \fW}\sum_{b',k,\sigma, \sigma'} \Pr(\nu_{\fO},\nu_{\fT}) \alpha_{k\sigma}\alpha^*_{k\sigma'}  \times \bigg\{ \\
& \quad  \bra{s_{\fO}+r_{\fO}}  \otimes \bra{r_{\fT}} \otimes \bra{b'} (X_{\fO}^{\fv} \otimes X_{\fT}^{\fw}) \times\\
& \quad \sigma  \left(\ketbra{s_{\fO} + r_{\fO}} \otimes \ketbra{r_{\fT}} \otimes \frac{\Id}{\Tr \Id} \right) \sigma' \times \\
& \quad  (X_{\fO}^{\fv} \otimes X_{\fT}^{\fw}) \ket{s_{\fO}+r_{\fO}} \otimes \ket{r_{\fT}}\otimes \ket{b'} \bigg\}
\end{align*}

Since the Pauli matrices are traceless, this imposes \(\sigma_l = \sigma'_l\) for \(l\notin \fO \cup \fT \), where subscript \(l\) is used to select the action of \(\sigma\) and \(\sigma'\) on qubit \(l\). For the output qubits, \(\sum_{r_{\fo}} \bra{s_{\fo}+r_{\fo}} X_{\fo}^{\fv_{\fo}} \sigma_{\fo} \ketbra{s_{\fo}+r_{\fo}} \sigma'_{\fo} X_{\fo}^{\fv_{\fo}} \ket{s_{\fo}+r_{\fo}}\) vanishes for \(\sigma_{\fo} \neq \sigma'_{\fo}\) and similarly for the traps, \(\sum_{r_{\ft}} \bra{r_{\ft}} X_{\ft}^{\fw_{\ft}} \sigma_{\ft} \ketbra{r_{\ft}} \sigma'_{\ft} X_{\ft}^{\fw_{\ft}} \ket{r_{\ft}}\) vanishes for \(\sigma_{\ft} \neq \sigma'_{\ft}\). Hence we get: 
\begin{align*}
\Pr(\mathrm{fail})
& \leq \sum_{\nu_{\fO},\nu_{\fT}}\sum_{\fv \in \fV, \fw \in \fW}\sum_{k,\sigma}\Pr(\nu_\fO,\nu_\fT) |\alpha_{k\sigma}|^2  \times \\
& \qquad \prod_{\fo \in \fO} |\bra{s_\fo+r_\fo} X_\fo^{\fv_\fo} \sigma_\fo \ket{s_\fo+r_\fo}|^2 \times \\
& \qquad \prod_{\ft\in \fT} |\bra{r_\ft} X_\ft^{\fw_\ft} \sigma_\ft \ket{r_\ft}|^2 \\
& \leq \sum_{k}\sum_{\sigma} |\alpha_{k\sigma}|^2 f(\sigma) 
\end{align*}
with 
\begin{align}\label{eq:f}
f(\sigma) =  & \sum_{\nu_\fO,\nu_\fT} \sum_{\fv\in \fV, \fw\in \fW}\Pr(\nu_\fO, \nu_\fT) \times  \nonumber \\
& \quad \prod_{\fo\in \fO} |\bra{s_\fo+r_\fo} X_\fo^{v_\fo} \sigma_\fo \ket{s_\fo+r_\fo}|^2 \times \nonumber \\
& \quad \prod_{\ft\in \fT} |\bra{r_\ft} X_\ft^{\fw_\ft} \sigma_\ft \ket{r_\ft}|^2 
\end{align}

\paragraph{Worst case scenario (for the upper-bound).}
The worst case scenario corresponds to maximizing the bound on \(\Pr(\mathrm{fail})\). Since we have \(\sum_{k, \sigma} |\alpha_{k\sigma}|^2 = 1\), our bound is worst when \(\alpha_{k\sigma} = 1\) for \(\sigma\) such that \(f(\sigma)\) is maximum. In this case we get: 
\[
\Pr(\mathrm{fail}) \leq \max_{\sigma} f(\sigma)
\]

\paragraph{Simplified expression for the bound.}
Given our protocol, a global trap and output qubit configuration \(\nu_\fO,\nu_\fT\) is defined by (i) the set \(\fT\) of trap qubits, itself entirely determined by the position and kind of test runs within the sequence of runs, and (ii) the preparation parameters \(\theta_l\) and \(r_l\) of each trap and output qubits. Each parameter of (i) and (ii) being chosen independently, the probability of a given configuration \(\nu_\fO, \nu_\fT\) can be decomposed into the probability \(\Pr(\fT)\) for a given configuration of trap locations multiplied by the probability of a given configuration for the prepared state of the trap and output qubits, \(\prod_{l\in \fO \cup \fT}  \sum_{\theta_l,r_l}\Pr(\theta_l, r_l)\). Using this, one can rewrite \(f(\sigma)\):
\begin{align*}
f(\sigma) = & \sum_{\fT} \sum_{\fv\in \fV, \fw\in \fW} \Pr(\fT) \times \\
& \quad \prod_{\fo\in \fO} \sum_{\theta_\fo, r_\fo} \Pr(\theta_\fo, r_\fo) |\bra{s_\fo + r_\fo} X_\fo^{v_\fo} \sigma_\fo \ket{s_\fo + r_\fo}|^2 \times \\
& \quad \prod_{\ft\in \fT} \sum_{\theta_\ft,r_\ft}\Pr(\theta_\ft, r_\ft) |\bra{r_\ft} X_\ft^{\fw_\ft} \sigma_\ft \ket{r_\ft}|^2 
\end{align*}

Now, let \(\sigma\) be a maximizing attack and denote by \(\sigma_{|X}\) the binary vector indexed by qubit positions of the combined computation where ones mark qubit positions for which \(\sigma\) acts as \(X\) or \(Y\). In the following, we allow \(\fO\) to also denote the  binary vector over qubit positions of the combined computation where ones are positioned for qubits in \(O\), and similarly for \(\fT\). Using the fact that \(|\bra{s_\fo+r_\fo} X_\fo^{\fv_\fo} \sigma_\fo \ket{s_\fo+r_\fo}|^2\) is 1 for \(X_\fo^{v_\fo} \sigma_\fo \in \{I,Z\}\) and 0 otherwise, we obtain that
\begin{align*}
& \prod_{\fo\in \fO}\sum_{\theta_\fo,r_\fo}\Pr(\theta_\fo, r_\fo) |\bra{s_\fo+r_\fo} X_\fo^{v_\fo} \sigma_\fo \ket{s_\fo+r_\fo}|^2 \\
& \quad =
\begin{cases}
1 & \mbox{for } \fO.\sigma_{|X} = \fv \\
0 & \mbox{otherwise} 
\end{cases}
\end{align*}
Where, for \(a\) and \(b\) binary vectors, \(a.b\) is the bit-wise binary product vector. We obtain a similar expression for the trap qubits. Inserting these expressions in Equation~\ref{eq:f} we obtain that for the attack to be successful, it must affect in a non-trivial way at least \(d/2\) computation runs, and at most \(w\) test runs. The probability of failure can thus be rewritten as: 
\[
\Pr(\mathrm{fail}) \leq \max_{\sigma} \sum_{\fT\in \Upsilon_\sigma} \Pr(\fT)
\]
where \(\Upsilon_\sigma\) are configurations where the binary vector \(\sigma_{|X}\) has ones in at most \(w\) test runs and has ones on at least \(d/2\) computation runs.

\paragraph{Closed form upper bound.}
Now, assume that the maximum above is attained for some $\sigma$ that happen to affect one of the run, say \(k\), on more than one qubit. Consider $\sigma'$ with the sole difference to $\sigma$ that only one of the qubits in run \(k\) is affected by the attack. Because run \(k\) is still non trivially affected by \(\sigma\) and \(\sigma'\), we conclude that all configurations \(\fT\) in \(\Upsilon_\sigma\) are also in \(\Upsilon_{\sigma'}\). Therefore 
\begin{align*}
	\Pr(\mathrm{fail}) \leq \max_m \max_{\sigma \in E_{m}} \sum_{T\in\Upsilon_\sigma} \Pr(T).
\end{align*}
where $E_{m}$ denotes the set of Pauli operators with $m$ single qubit deviations all in distinct runs. Note that the parameter $m$ and the locations of the attacks within each run describe the adversary's strategy.

Additionally, since the random choice of test runs is completely uniform, the term $\sum_{T\in\Upsilon_\sigma} \Pr(T)$ is invariant under permutations of the test and computation runs. We can hence restrict the range of the maximum to the specific Pauli operators $\sigma_{m}$ with a deviation on a single qubit in each of the first $m$ runs:
\begin{align}
	\Pr(\mathrm{fail}) \leq \max_m \sum_{T\in\Upsilon_{\sigma_{m}}} \Pr(T).\label{eq:upper_bound}
\end{align}

\begin{figure*}[ht]
  \centering
  \begin{minipage}{0.9\textwidth}
  	\includegraphics[]{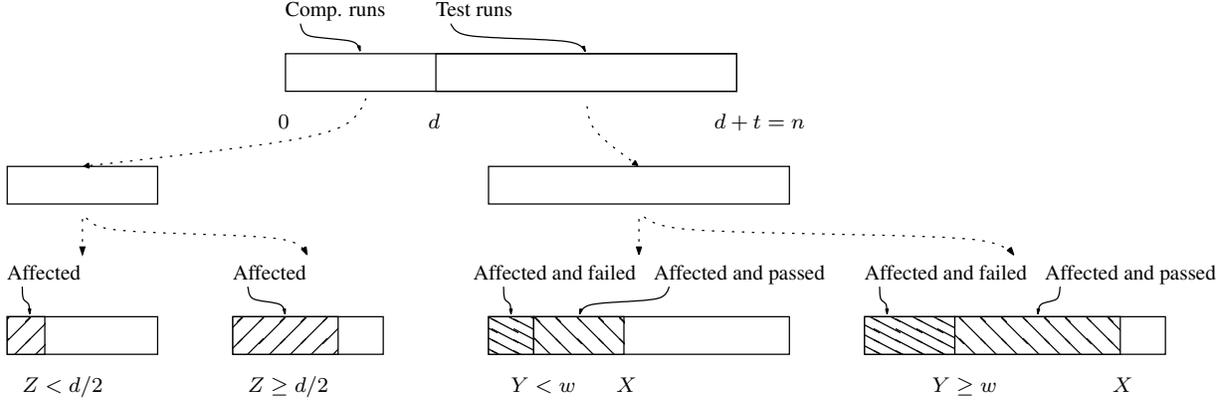}
  \end{minipage}
  \caption{The four cases needed to determine a closed form upper bound for the probabiliity of failure. First, we determine the probability for the number of affected computation runs. If it is low enough ($Z < d/2$), no need to abort. If it is high ($Z \geq d/2$), we find a bound on the probability that the number of failed test runs $Y$ is below or above $w$.}
  \label{fig:vbqc}
\end{figure*}

\paragraph*{Formal bound.}
To find a closed form upper bound for the right-hand side of Equation~\ref{eq:upper_bound} we distinguish two regimes (see Figure~\ref{fig:vbqc}):
\begin{enumerate}
	\item For $m \geq \left( \frac{1}{2} - \varphi \right) n$ with some (small) $\varphi > 0$, we upper the right-hand side with the probability that $\sigma_m$ triggers not more than $w$ traps, thus ignoring the condition that \(\sigma_m\) affects more than \(d/2\) computation runs. This again is done in two steps:
		\begin{enumerate}
			\item Find a lower bound (with high probability) on the number of affected test runs.
			\item Use this number to find a lower bound on the number of triggered traps.
		\end{enumerate}
	\item For $m \leq \left( \frac{1}{2} - \varphi \right) n$, we show an upper bound on the probability that $\sigma_m$ manages to affect at least $d/2$ computation runs, then ignoring the condition that \(\sigma_m\) affects less than \(w\) traps.
\end{enumerate}

As a first step, we derive a bound for the probability of the random variable $X$ counting the test runs affected by the server's deviation to be smaller than a given threshold $\frac{m}{n} -\epsilon_1$. Because $X$ is $(n, t, m)$-hypergeometrically distributed, for all  $\varepsilon_1 > 0$, we can apply Corollary~\ref{cor:lower_tail} to obtain:
\begin{align}
	\Pr \left[ X \leq \left( \frac{m}{n} - \varepsilon_1 \right) t \right] \leq \exp \left( -2 \frac{t^2}{m} \varepsilon_1^2 \right). \label{eqn:upper_bound_X}
\end{align}
In other words, this inequality means that with high probability, the attack will affect at least $\left( \frac{m}{n} - \varepsilon_1 \right) t$ test runs.

As the next step, we can derive from here a bound on the probability that the random variable $Y$ is below some threshold, where $Y$ describes the number of failed test runs, i.e. the number of affected test runs where the deviation hits a trap. Since the type of test runs is sampled independently from the location of test runs, $Y$ conditioned on the event that $X = \left( \frac{m}{n} - \varepsilon_1 \right) t$ follows a $(\left( \frac{m}{n} - \varepsilon_1 \right) t, 1/k)$-binomial distribution. Let $\varepsilon_2 > 0$. Applying Lemma~\ref{lemma:hoeffding_binomial}, we arrive at
\begin{align*}
	&\Pr \left[ \left. Y \leq \left( \frac{1}{k} - \varepsilon_2 \right) \left( \frac{m}{n} - \varepsilon_1 \right) t \,\right|\, X = \left( \frac{m}{n} - \varepsilon_1 \right) t \right] \\
	&\leq \exp \left( -2t \left( \frac{m}{n} - \varepsilon_1 \right) \varepsilon_2^2 \right).
\end{align*}
Combining the previous expressions, we obtain
\begin{align*}
	&\Pr \left[ Y \leq \left( \frac{1}{k} - \varepsilon_2 \right) \left( \frac{m}{n} - \varepsilon_1 \right) t \right] \\
	&= \Pr \left[ \left. Y \leq \left( \frac{1}{k} - \varepsilon_2 \right) \left( \frac{m}{n} - \varepsilon_1 \right) t \,\right|\, X < \left( \frac{m}{n} - \varepsilon_1 \right) t \right] \\
	&\qquad \cdot \Pr \left[ X < \left( \frac{m}{n} - \varepsilon_1 \right) t \right] \\
	&\quad + \Pr \left[ \left. Y \leq \left( \frac{1}{k} - \varepsilon_2 \right) \left( \frac{m}{n} - \varepsilon_1 \right) t \,\right|\, X \geq \left( \frac{m}{n} - \varepsilon_1 \right) t \right] \\
	&\qquad \cdot \Pr \left[ X \geq \left( \frac{m}{n} - \varepsilon_1 \right) t \right]. \\
\end{align*}
Note that decreasing $X$ also makes $Y$ smaller. This is to be understood in the following way: $Y$ is dependent on $X$. For all $x_1 \leq x_2$ and for all $y$ it holds that $\Pr [ Y \leq y | X = x_1 ] \geq \Pr [ Y \leq y | X = x_2 ]$. In other words, $Y$ conditioned on $X=x_1$ is less than $Y$ conditioned on $X=x_2$ in the \emph{usual stochastic order}. Applying Inequality~\eqref{eqn:upper_bound_X} and exploiting this fact, we continue upper-bounding the previous expression by
\begin{align*}
	&\leq \exp \left( -2 \frac{t^2}{m} \varepsilon_1^2 \right) \\
	&\quad + \Pr \left[ \left. Y \leq \left( \frac{1}{k} - \varepsilon_2 \right) \left( \frac{m}{n} - \varepsilon_1 \right) t \,\right|\, X = \left( \frac{m}{n} - \varepsilon_1 \right) t \right] \\
	&\leq \exp \left( -2 \frac{t^2}{m} \varepsilon_1^2 \right) + \exp \left( -2t \left( \frac{m}{n} - \varepsilon_1 \right) \varepsilon_2^2 \right).
\end{align*}
As the final step, let $Z$ be a random variable describing the number affected computation runs. Note that $Z$ is $(n, d, m)$-hypergeometrically distributed. We are then interested in the following probability which we can upper-bounded using Corollary~\ref{cor:upper_tail} for $\frac{m}{n} < \frac{1}{2}$:
\begin{align*}
	\Pr \left[ Z \geq \frac{d}{2} \right] \leq \exp \left( -2 \left( \frac{1}{2} - \frac{m}{n} \right)^2 \frac{d^2}{m} \right).
\end{align*}
To compute the probability of failure, we equate
\begin{align*}
	w = \left( \frac{1}{k} - \varepsilon_2 \right) \left( \frac{m}{n} - \varepsilon_1 \right) t
\end{align*}
and conclude that
\begin{align*}
	&\Pr \left[ \mathrm{fail} \right] \leq \max_m \sum_{T\in\Upsilon_{\sigma_{m}}} \Pr(T) = \max_m \Pr \left[ Y \leq w \, \wedge \, Z \geq \frac{d}{2} \right] \\
	&\leq \max \left\{ \max_{m \leq \left( \frac{1}{2} - \varphi \right)n} \Pr \left[ Z \geq \frac{d}{2} \right], \max_{m \geq \left( \frac{1}{2} - \varphi \right) n} \Pr \left[ Y \leq w \right] \right\}.
\end{align*}
Both inner maximums are attained for $m = \left( \frac{1}{2} - \varphi \right)n$. Hence, we can upper-bound the preceding terms by
\begin{align}
	&\leq \max \left\{ \exp \left( -2 \cdot \frac{\varphi^2}{1/2-\varphi} \cdot \frac{d^2}{n} \right) ,\, \right. \nonumber \\
	&\left. \exp \left( -2 \frac{t^2}{\left( \frac{1}{2} - \varphi \right)n} \varepsilon_1^2 \right) + \exp \left( -2t \left( \frac{1}{2} - \varphi - \varepsilon_1 \right) \varepsilon_2^2 \right) \right\} \label{eq:bound},
\end{align}
where
\begin{align*}
	w/t = \left( \frac{1}{k} - \varepsilon_2 \right) \left( \frac{1}{2} - \varphi - \varepsilon_1 \right),
\end{align*}
and
\begin{align*}
	&0 < \varepsilon_1 < \frac{1}{2}, \\
	&0 < \varepsilon_2 < \frac{1}{k}, \\
	&0 < \varphi < \frac{1}{2} - \varepsilon_1.
\end{align*}
To obtain an optimal bound, this expression must be minimized over $\varepsilon_1$, $\varepsilon_2$, and $\varphi$.

Irrespective of the exact form of the optimal bound, choosing $\varphi$, $\varepsilon_1$, and $\varepsilon_2$ sufficiently small implies the existence of protocols with verification exponential in $n$, for any fixed $0 < w/t < \frac{1}{2k}$ and fixed $\frac{d}{n}, \frac{t}{n} \in (0,1)$.
\end{proof}

%% file: sec_files/app_proof_rob.tex
In the following, we define the constant ratios of test, computation and tolerated failed test runs as $\delta := d/n$, $\tau := t/n$ and $\omega := w/t$.

\begin{theorem}[Local-Correctness of VDQC Protocol on Noisy Devices]
\label{thm:correctness_form}
Assume a Markovian round-dependent model for the noise on Client and Server devices and let $\pmin \leq \pmax < 1/2$ be respectively a lower and an upper-bound on the probability that at least one of the trap measurement outcomes in a single test round is incorrect.

If $\omega > \pmax$, Protocol~\ref{prot:MQ-VBQC} is $\epsilon_{\mathit{cor}}$-locally-correct with exponentially-low $\epsilon_{\mathit{cor}}$:
\begin{align*}
	\epsilon_{\mathit{cor}} = &\exp \left( -2 (\omega - \pmax)^2 \tau n \right) \\
	&\quad+ \exp \left( -2 \left(\frac{1}{2} - \pmax \right)^2 \delta n \right).
\end{align*}

On the other hand, if $\omega < \pmin$, then the Client's acceptance probability in Protocol~\ref{prot:MQ-VBQC} is exponentially-low $\exp \left( -2 (\pmin - \omega)^2 \tau n \right)$.

\end{theorem}

\begin{proof}

We define random variables $Y$ that corresponds to the number of failed test runs during one execution of the protocol, and $Z$ counting the number of affected computation runs (where at least one bit of output is flipped). We call $\ok$ the event that the Client accepts at the end of the protocol - if not many test runs fail, meaning that $Y < w$ - and $\Correct$ the event corresponding to a correct output - if few of the computation runs have their output bits flipped and therefore $Z < d/2$.

\paragraph{For $\omega > \pmax$.} Equivalently, we have that $w > t\pmax$. We are looking to lower-bound the probability of an honest run producing the correct outcome and not aborting:

\begin{align*}
\Pr \left[\Correct \land \ok \right] &= \Pr \left[ Z < \frac{d}{2} \, \land \, Y < w \right] \\
&= \Pr \left[ Z < \frac{d}{2}\right] \Pr \left[ Y < w \right]
\end{align*}

The second equality stems from the fact that the noise on runs is independent across runs and the nature of each run is chosen uniformly at random. We start by considering the effect on computation runs. If a computation run is affected by a given noise, then there is at least one type of test run that would have been affected (triggering traps) by the same noise. Since $\pmax$ is an upper-bound on the probability that any type of test run fails, then it is also an upper-bound on the probability that the outcome of the computation is incorrect. Let $\hat{Z}_1$ be a random variable following a $(d, \pmax)$-binomial distribution. Since we suppose that the noise is not correlated across runs (meaning that the probabilities runs fail on noisy devices are independent), $Z$ is upper-bounded by $\hat{Z}_1$ in the \emph{usual stochastic order}, which in this case gives:

\begin{align*}
\Pr \left[ Z < \frac{d}{2} \right] \geq &\Pr \left[ \hat{Z}_1 < \frac{d}{2} \right] = 1 - \Pr \left[ \hat{Z}_1 > \frac{d}{2} \right]
\end{align*}

Since $\pmax < \frac{1}{2}$ then $E\left (\hat{Z}_1 \right) = d\pmax < \frac{d}{2}$ and Lemma~\ref{lemma:hoeffding_binomial} yields:

\begin{align*}
\Pr \left[ \hat{Z}_1 > \frac{d}{2} \right] &\leq \exp \left( -2 \frac{\left( d \pmax - \frac{d}{2} \right)^2}{d} \right)\\
&= \exp \left( -2\delta \left( \pmax - \frac{1}{2} \right)^2 n \right) = \epsilon_Z \in \operatorname{negl}(n)
\end{align*}

Then $\Pr [\Correct] = \Pr \left[ Z < \frac{d}{2} \right] \geq 1 - \epsilon_Z$.

We can now focus on the test runs. Note that $Y$ describes exactly the number of test rounds in which at least one trap measurement outcome is incorrect (by definition of a failed test run). The probability that a given test run fails is therefore upper-bounded by $\pmax$. Let $\hat{Y}_1$ be a random variable following a $(t, \pmax)$-binomial distribution. Since we suppose that the noise is not correlated across runs, $Y$ is upper-bounded by $\hat{Y}_1$ in the usual stochastic order:

\begin{align*}
\Pr \left[ Y < w \right] \geq &\Pr \left[ \hat{Y}_1 < w \right] = 1 - \Pr \left[ \hat{Y}_1 > w \right]
\end{align*}

Further, since $E\left(\hat{Y}_1 \right) = t\pmax < w$, applying Lemma~\ref{lemma:hoeffding_binomial} yields:
\begin{align*}
\Pr \left[ \hat{Y}_1 > w \right] \leq &\exp \left( -2 \frac{(t \pmax - w)^2}{t} \right)\\
&= \exp \left( -2 (\omega - \pmax)^2 \tau n \right) = \epsilon_{Y, 1}
\end{align*}

Then $\Pr [\ok] = \Pr \left[ Y < w \right] \geq 1 - \Pr \left[ \hat{Y}_1 > w \right] = 1 - \epsilon_{Y, 1}$.

Combining these inequalities gives:

\begin{align*}
\Pr \left[\Correct \land \ok \right] \geq (1 - \epsilon_Z)(1 - \epsilon_{Y, 1}) \geq 1 - (\epsilon_Z + \epsilon_{Y, 1})
\end{align*}

\paragraph{For $\omega < \pmin$.} In that case, we have that $w < t\pmin$. We show that the probability of accepting is upper-bounded by a negligible function. 
Let $\hat{Y}_2$ be a random variable following a $(t, \pmin)$-binomial distribution, $Y$ then is lower-bounded by $\hat{Y}_2$ in the usual stochastic order:

\begin{align*}
\Pr \left[ Y < w \right] \leq \Pr \left[ \hat{Y}_2 < w \right]
\end{align*}

Since $w < t\pmin$, using Lemma~\ref{lemma:hoeffding_binomial} directly and the same simplifications as above, we get:

\begin{align*}
\Pr \left[ \hat{Y}_2 < w \right] \leq \exp \left( -2 (\pmin - \omega)^2 \tau n \right) = \epsilon_{Y, 2}
\end{align*}

Therefore $\Pr \left[ \ok \right] \leq \epsilon_{Y, 2} \in \operatorname{negl}(n)$.

\end{proof}